\documentclass[acmsmall]{acmart}
\usepackage{graphicx}
\usepackage[free-standing-units=true]{siunitx}
\usepackage{amsfonts}
\usepackage{amsthm}
\usepackage{amsmath}
\usepackage{amscd}
\usepackage{bbold}
\usepackage{bussproofs}
\usepackage[latin2]{inputenc}
\usepackage{t1enc}
\usepackage[mathscr]{eucal}
\usepackage{indentfirst}
\usepackage{fancyvrb}
\usepackage{graphicx}
\usepackage{graphics}
\usepackage{mathpartir}
\usepackage{pict2e}
\usepackage{epic}
\usepackage{syntax}
\usepackage{enumitem}
\usepackage{epstopdf}
\usepackage[labelfont=bf]{caption}

\usepackage[switch]{lineno}
\usepackage{tikz-cd}
\usepackage{multicol}
\usepackage{pifont}

\usepackage{ stmaryrd }

\usepackage{mdframed}
\usepackage{xcolor}
\usepackage{tcolorbox}

\definecolor{cornellred}{RGB}{196,18,48}
\definecolor{dartmouthgreen}{RGB}{0,112,60}
\definecolor{cisorange}{RGB}{247,147,33}
\definecolor{cisblue}{RGB}{67,199,244}

\newcommand{\cat}{%
\mathbf%
}

\newcommand{\sem}[1]{
  \left\llbracket #1 \right\rrbracket
}

\newcommand{\bsem}[1]{
  \llparenthesis #1 \rrparenthesis
}

\newcommand{\lettext}[0]{\text{let }}
\newcommand{\lift}[0]{\mathsf{lift}}
\newcommand{\iter}[0]{\text{it }}
\newcommand{\D}{\mathcal{D}}

\newcommand{\dint}{\mathop{}\! \mathrm{d}}

\theoremstyle{plain}
\newtheorem{Th}{Theorem}[section]
\newtheorem{Lemma}[Th]{Lemma}

\theoremstyle{definition}
\newtheorem{Def}[Th]{Definition}

\newtheorem{Rem}[Th]{Remark}
\newtheorem{?}[Th]{Problem}
\newtheorem{Ex}[Th]{Example}

\newcommand{\nat}{\mathbb{N}}
\newcommand{\R}{\mathbb{R}}
\newcommand{\B}{\mathbb{B}}

\newcommand{\test}[1]{\mathcal D(#1)}
\newcommand{\dist}[1]{\mathcal D'(#1)}
\newcommand{\diff}[1]{(#1, \mathcal P^U_{#1})}
\newcommand{\ifthen}[3]{\mathsf{if }\, #1\, \mathsf{ then }\, #2\, \mathsf{ else }\, #3}
\newcommand{\lamb}[2]{\lambda #1.\, #2}
\newcommand{\app}[2]{#1\, #2}
\newcommand{\letin}[3]{\mathsf{let}\, #1 \, = \, #2 \, \mathsf{in} \, #3}

\newcommand{\diffcat}{\cat{Diff}}

\usepackage{pgfplots}

\newcommand{\set}[2]{\{#1 \, | \, #2 \}}

\newcommand{\lrang}[1]{ \langle #1 \rangle }

\setcopyright{none}

\usepackage{booktabs}   
\usepackage{subcaption} 

\begin{document}

\title{Distribution Theoretic Semantics for Non-Smooth Differentiable Programming}         


\author{Pedro H. Azevedo de Amorim}
\affiliation{
  \department{Computer Science}              
  \institution{Cornell University}            
  \city{Ithaca}
  \state{New York}
  \postcode{14850}
  \country{United States}                    
}
\email{pamorim@cs.cornell.edu}          

\author{Christopher Lam}
\affiliation{
  \department{Computer Science}             
  \institution{University of Illinois at Urbana-Champaign}           
  \city{Champaign}
  \state{Illinois}
  \postcode{61802}
  \country{United States}                   
}
\email{lam30@illinois.edu}         

\begin{abstract}
  With the wide spread of deep learning and gradient descent inspired optimization algorithms, differentiable programming has gained traction. Nowadays it has found applications in many different areas as well, such as scientific computing, robotics, computer graphics and others. One of its notoriously difficult problems consists in interpreting programs that are not differentiable everywhere.
  
  In this work we define $\lambda_\delta$, a core calculus for non-smooth differentiable programs and define its semantics using concepts from distribution theory, a well-established area of functional analysis. We also show how $\lambda_\delta$ presents better equational properties than other existing semantics and use our semantics to reason about a simplified ray tracing algorithm. Further, we relate our semantics to existing differentiable languages by providing translations to and from other existing differentiable semantic models. Finally, we provide a proof-of-concept implementation in PyTorch of the novel constructions in this paper.
\end{abstract}



\maketitle

\section{Introduction}
The field of differentiable programming languages seeks to add differentiation operators, such as gradients, to programming languages. Originally motivated by its use in machine learning frameworks such as Tensorflow \cite{tensorflow}, these techniques have seen a considerable rise in interest in the past years.

In the context of neural networks, differentiable programming is used to easily implement gradient descent algorithms. It can quickly find the local minimum of a carefully selected objective function over a large data set, yet can be succinctly represented by the following two line program:
\begin{align*}
    w_0 = &\ 0\\
    w_{n + 1} = &\ w_n - \gamma \nabla F(w_n)
\end{align*}

Many languages can already easily implement the program above for a broad class of objective functions $F$, but in doing so they have outpaced the theoretical understanding of differentiable programming. This is an expected consequence, as the gradient of a function is only defined for points in which the function is differentiable. Unfortunately, it is very common for real-world applications to make use of non-smooth functions. The ReLU function depicted in Figure~\ref{fig:relu}, for instance is a non-smooth function widely used to train neural networks.


The example above illustrates how the practice of differentiable programming has outpaced its theory. Much has been done in understanding syntactic aspects of differentiation in programming languages; automatic differentiation (AD) techniques have seen steady progress since the 80s \cite{griewank1989automatic,pearlmutter2008,beck1994if}. On the denotational side, however, the programming languages community lags behind in handling features that users of these languages take for granted.

This disparity has practical and negative consequences as many automatic differentiation algorithms present unsound behavior that breaks the equational theory of the language when the program being differentiated uses certain seemingly harmless features. 

One of these features is conditionals, which can easily create points of non-differentiability as the ReLU activation function shows:

    $$\verb_ ReLU(x) = if x < 0 then 0 else x_$$
    
    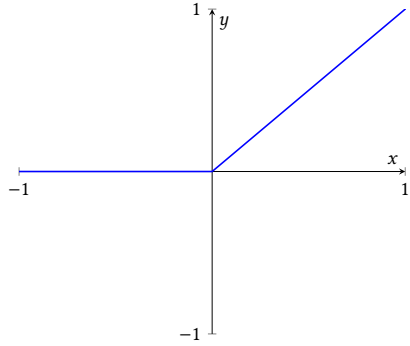
\begin{figure}
      \centering
      \scalebox{.75}{
        \begin{tikzpicture}[
          declare function={
            func(\x) = (x <= 0)*(0) +
            (\x > 0) * x
            ;
          }
          ]
          \begin{axis}[
            axis x line=middle, axis y line=middle,
            ymin=-1, ymax=1, ytick={-1,...,1}, ylabel=$y$,
            xmin=-1, xmax=1, xtick={-1,...,1}, xlabel=$x$,
            domain=-1.5:1.5,samples=101, 
            ]
            \addplot [blue,thick] {func(x)};
          \end{axis}
        \end{tikzpicture}
      }
      \caption{ReLU activation function}
      \label{fig:relu}
\end{figure}

One way of fixing this is by syntactically restricting which programs can be differentiated. This is an insufficient solution, as many common applications heavily rely on these use-cases. 

Many attempts have been made to deal with non-smooth differentiability, but they all either present certain non-intuitive restrictions on which non-smooth programs can be written or break the equational theory of the language. Concretely, many of the expected equations that if-statements are supposed to satisfy are not validated by these semantics, such as:

\begin{Th}
  \label{th:intro}
   $\sem{\mathsf{if }\, b\, \mathsf{ then \,  } t \, \mathsf{ else }\, t} = \sem {t}$.
\end{Th}

Our work uses ideas from distribution theory ---  a subject relevant to functional analysis, partial differential equations and mathematical physics --- to give a novel approach to non-smooth differentiability where conditionals have a more intuitive equational theory. Our semantics is the first one that validates Theorem~\ref{th:intro}\footnote{In $\lambda_\delta$ syntax the theorem statement looks a bit different, see Section~\ref{ref:prac} for more details.} while still allowing for a large set of boolean predicates that can be conditioned on. Furthermore, as we will show in Section~\ref{sec:examples}, the existing semantics for non-smooth differentiation behave unsoundly when you have the interaction of integration and differentiation. These problematic interactions are both common in fields such as graphics, robotics, and physical simulators, and also completely mathematically sound in our semantics. 

In this work we present $\lambda_\delta$, a simply-typed $\lambda$-calculus extended with constructions motivated by distribution theory, define its denotational semantics using diffeological spaces (see \ref{sec:diffspace}), and implement a proof-of-concept library in PyTorch.

We chose to use diffeological spaces because they are well-understood mathematical objects that can accommodate higher-order programming, differentiation and distribution-theoretic ideas. However, we believe that there might be other interesting semantical models for $\lambda_\delta$, specially when taking into account computability issues; see Section~\ref{ref:prac} for more details.

This paper assumes significant background in distribution theory, and uses the standard notation from this field. See Appendix~\ref{sec:background} for a self-contained introduction on this topic. We heavily recommend reading this appendix for those that are unfamiliar with this material, as we will be making heavy use of it in our contributions. \\

\noindent\textbf{Summary of contributions.} $\lambda_\delta$ is the first language that combines conditionals, higher-order functions, differentiability, and has a denotational semantics. Our main contributions are as following:

\begin{itemize}
\item We define $\lambda_\delta$, an extension of the simply-typed $\lambda$-calculus with differentiation and distribution-theoretic primitives.(\S \ref{sec:syntax})
\item We use the category $\diffcat$ to interpret its denotational semantics and prove Theorem~\ref{th:main}.(\S \ref{sec:sem})
\item We use our semantics to reason about several non-smooth programs, such as a modern differentiable ray tracing algorithm. (\S \ref{sec:validray})
\item We embed a smooth $\lambda$-calculus of \cite{huot2020} in $\lambda_\delta$ and show how to relate their AD transform with our syntax and semantics. (\S \ref{sec:embed})
\item We propose a constructive semantics for $\lambda_\delta$ using ideas from constructive topology\cite{sherman2019,sherman2020}.(\S \ref{sec:embed})
\item We have implemented a proof-of-concept library in PyTorch which explores the ideas presented in this work. (\S \ref{sec:impl})
\end{itemize}

\section{Mathematical Preliminaries: Diffeological Spaces}
\label{sec:diffspace}

When dealing with differentiable programming it is expected that the category in which we interpret our language has a canonical notion of smoothness. A well-known such category is $\cat{Man}$, the category of manifolds and smooth maps. Even though manifolds are extremely well-understood at this point, $\cat{Man}$ cannot be used to interpret a smooth $\lambda$-calculus, as it is not cartesian closed, meaning that it cannot interpret higher-order programs.

This limitation was already noted by mathematicians which led to the development of diffeological spaces. \footnote{c.f. the Diffeology textbook \cite{diffeology} for a presentation on the subject} Though its definition might seem foreign to someone who is only used to working with smooth manifolds, diffeological spaces are deeply related to them, as they can be defined as a certain category of sheaves over $\cat{Man}$.

The soon to be defined category $\cat{Diff}$ has many desired categorical properties such as completeness, cocompleteness and cartesian closure; indeed, it is even a \emph{quasitopos}.

\begin{Def}
  A diffeological space is a pair $\diff X$, where $X$ is a set and for every natural number $n$ and open set $U \subseteq \R^n$, a set of plots $\mathcal{P}^U_X \subseteq U \to X$ such that.
  \begin{itemize}
  \item Every constant function is a plot
  \item If $f : V \to U$ is a smooth function and $p \in \mathcal P^U_X$ then $p \circ f \in \mathcal P^V_X$
  \item Let $f : U \to X$ be a function, if for every point $u \in U$ there exists an open set $V \subseteq U$ such that $f|_{V} \in \mathcal P^V_X$ then $f \in \mathcal P^U_X$ 
  \end{itemize}
\end{Def}

We say that a function $f : \diff X \to \diff Y$ is a $\cat{Diff}$ morphism if for every plot $p : U \to X \in \mathcal P^U_X$, $f \circ p \in \mathcal P^U_Y$. We call the set $\mathcal P^U_X$ a diffeology.

\begin{Def}
  The category $\cat{Diff}$ has diffeological spaces as morphisms and $\cat{Diff}$ morphisms as arrows.
\end{Def}

Note that this construction is very similar to the quasi Borel category for higher-order probability theory defined by Heunen et al. \cite{qbs}. Unsurprisingly, they both arise from similar categorical machinery.

Let us work through a couple of examples to get a better feel of how to work with diffeological spaces.

\begin{Ex}
  The pair $(\R^n, \set {f : U \to \R^n} {f \text{ is smooth})})$ is a diffeological space for every $n > 0$.
\end{Ex}

The example above shows that the common definition of smoothness is used when defining the $\R$ diffeological space. This construction can be generalized to an arbitrary manifold.

\begin{Ex}
  Let $M$ be a manifold, then $(M, \set {f : U \to M} {f \in \cat{Man}(U, M)})$ is a diffeological space. Every smooth function $f : M \to N$ between manifolds is also a $\cat{Diff}$ morphism. This construction is actually a functor $\iota : \cat{Man} \to \cat{Diff}$.
\end{Ex}

\begin{Lemma}[\cite{diffeology}]
  The functor $\iota$ is full and faithful.
\end{Lemma}

What the theorem above implies is that when you only have ground types the smooth functions are exactly what you would expect them to be. In particular $\cat{Man}(\R^n, \R^m) = \cat{Diff}(\R^n, \R^m)$. This is what makes $\cat{Diff}$ such a nice category to interpret differentiable programs --- it simply conservatively extends the familiar category $\cat{Man}$.

\begin{Ex}
  The pair $\diff 1$ is a diffeological space where $\mathcal{P}^U_1$ is the singleton set for every open $U$.
\end{Ex}

\begin{Lemma}
\label{lem:const}
  Let $X$ be a set, the set $\set{f : U \to \R}{f \text{ is constant}}$
  is a diffeology.
\end{Lemma}
\begin{proof}
     The proof follows by unfolding the definition of diffeology. This is called the \emph{constant} diffeology.
\end{proof}

\begin{lemma}
    \label{lem:constmor}
    Let $X$ and $Y$ be diffeological spaces equipped with the constant diffeology. Every function $f : X \to Y$ is a $\cat{Diff}$ morphism.
\end{lemma}
\begin{proof}
    For every plot $p : U \to X$, the function $f \circ p$ is constant.
\end{proof}

\subsubsection{Cartesian Closed structure}

As it was mentioned above, in order to interpret higher-order programs we need $\cat{Diff}$ to be cartesian closed. 

\begin{Def}[Products]
  Let $\diff {X_1}$ and $\diff {X_2}$ be two diffeological spaces, then $\diff {X_1 \times X_2}$ is a cartesian product in $\cat{Diff}$, where $\mathcal P^U_{X_1\times X_2} = \set{p : U \to X_1 \times X_2}{p \circ \pi_i \in \mathcal{P}^U_{X_i}, i \in \{ 1, 2 \}}$
\end{Def}

\begin{Def}[Closure]
  Let $\diff X$ and $\diff Y$ be two diffeological spaces, then $\diff {X \Rightarrow Y}$ is an internal hom in $\cat{Diff}$, where $\mathcal P^U_{X\Rightarrow Y} = \set{p : U \to \cat{Diff}(X, Y)}{ (u, x) \mapsto p(u)(x) \in \cat{Diff}(U \times X, Y)}$
\end{Def}

These constructions will be used in Section~\ref{sec:sem} to interpret product types and function types.

\subsubsection{Distribution theory in $\cat{Diff}$}



There are other cartesian closed categories that also have a notion of differentiability. However, we are working with $\cat{Diff}$ because it accommodates the distribution theoretic machinery we need in order to define our semantics. The following lemma is used to define distribution objects inside $\cat{Diff}$.

\begin{Lemma}[\cite{diffeology}]
  Let $\diff X$ be a diffeological space and $Y \subseteq X$. The pair $(Y, \mathcal{P}^U_{X}|_Y)$ is a diffeological space, where $\mathcal{P}^U_{X}|_Y = \set{f : U \to X \in \mathcal{P}^U_{X}}{f(U) \subseteq Y}$.
\end{Lemma}

For the categorically minded reader: since $\cat{Diff}$ is a quasitopos, the theorem above can be generalized to arbitrary strong monomorphims.

This theorem makes it easy to equip $\test {\R^n}$ and $\dist{\R^n}$ with diffeologies: they are simply the appropriate subobjects of $\R^n \Rightarrow \R$ and $\test{\R^n} \Rightarrow \R$, respectively. Concretely, the plots in $\test{\R^n}$ are the functions $p : U \to \test{\R^n}$ such that $p$ is a plot in $\mathcal{P}^U_{\R^n \Rightarrow \R}$.

The theorem below is what allows us to lift smooth programs to distributions

\begin{Lemma}[\cite{kock2004}]
  \label{th:smooth}
  The map $\mathbf{T}: C^\infty(\R^n) \to \mathcal D'(\R^n)$ is smooth
\end{Lemma}

The next theorem is what allows us to interpret the $\delta$ distribution.

\begin{Th}
\label{th:dirac}
  The map $\delta : \R^n \to \dist{\R^n}$ is smooth.
\end{Th}
\begin{proof}
  The map $\eta : \R^n \to (\test{\R^n} \Rightarrow \R)$ such that $\eta(r, f) = f(r)$ is smooth by the Cartesian closed structure of $\cat{Diff}$ and we can conclude that $\delta$ is smooth by the fact that $\eta(r)$ is linear and, therefore, an element of $\dist{\R^n}$
\end{proof}

\begin{Rem}
It is important to note that the idea of using distribution theory to define semantics of differentiable programming languages does not rely on an specific category. We chose $\cat{Diff}$ for the sake of convenience. The category of convenient vector spaces \cite{blute2010}, for instance, also has objects for distributions. This means that the ideas presented in this paper may be used in other categories as well. 
\end{Rem}

\section{Background}
\label{sec:ifstmt}

It is not hard to see that if-statements can easily define non-differentiable behavior.
Before we explain the subtleties of conditionals in differentiable programming languages we motivate non-smooth differentiable programming languages. We argue that making it impossible to define non-smooth programs may introduce unnecessary complexity to models while allowing non-smooth behavior may also simplify certain models. 

Something that commonly occurs in systems that make use of differentiable components is having a model that, modulo its points of non-differentiability, works as intended. A common next step is to come up with a smooth variant for it, the sigmoid $S(x) = \frac{1}{1 + e^{-x}}$function being a smooth variant of the ReLU function.  Unfortunately, in more convoluted cases the smooth approximation does not have a nice closed-form expression, making the model more complex. Besides, as we will see in Section~\ref{sec:why}, in many cases, the tools used to make the model smooth implicitly use ideas from distribution theory. 

On the other hand, in physics, non-smoothness has been historically used as a simplifying agent. In rigid body physics, for instance, it is standard to assume that collisions occur instantaneously or that electrical charges are point-mass. With differentiable programming being used in physical simulators \cite{difftaichi} it seems natural that differentiable programming tools should accommodate these commonly used modeling principles.

What the existing semantics of non-smooth differentiability show us is that we cannot rely on all of our intuitions about differentiability. One of their drawbacks is that, while they prove that their semantics has the expected behavior modulo a null measure set, they break the equational theory of conditionals and add unnecessary non-terminating behavior to programs.

Our approach is more powerful in comparison; we offer a semantics that both allows for conditional statements as well as a method of differentiating those conditional statements without introducing undefined behavior. We do this by introducing distributions into our semantics and syntax. 

As differentiable programming finds applications in fields other than machine learning, such as robotics, computer vision, computer graphics and scientific computing \cite{degraverobot,LiImageProcessing,innesscicomp,PhysRevLettPhysSim}, we have empirical evidence that non-smoothness may result in better models. Recently Li et al. \cite{li2018} have shown how by using ideas from distributions theory they were able to come up with a  better differentiable ray-tracing algorithm.

Next, we will use the ReLU function to present some of the subtleties that if-statements introduce in differentiable programming languages. ReLU is differentiable almost everywhere, except at $x = 0$.

When we differentiate this function while ignoring this problematic point we get something that looks like the Heaviside function depicted in Figure~\ref{fig:heaviside}. This means that any solution that wants to accommodate non-smooth but continuous functions and higher-order derivatives must also deal with discontinuous functions. 



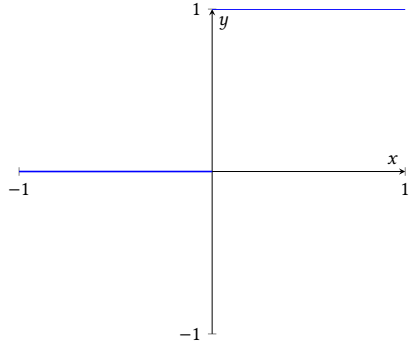
\begin{figure}
    \scalebox{.75}{
      \begin{tikzpicture}
        \begin{axis}[
          axis x line=middle, axis y line=middle,
          ymin=-1, ymax=1, ytick={-1,...,1}, ylabel=$y$,
          xmin=-1, xmax=1, xtick={-1,...,1}, xlabel=$x$,
          domain=-1.5:1.5,samples=101, 
          ]
          \addplot [domain=-1:0,blue,thick] {0};
          \addplot [domain=0:1,blue,thick] {1};
        \end{axis}
      \end{tikzpicture}
    }
    \caption{Heaviside function}
    \label{fig:heaviside}
\end{figure}

To our knowledge, there are two approaches to deal with differentiability of if-statements. The first is simply removing conditionals from the language and only allowing for the construction of infinitely differentiable functions. While this approach is formally valid, it loses the ability to construct the widely used ReLU activation function, severely limiting the expressiveness of the language. The second approach is somewhat more interesting: differentiate down both branches and leave the derivative as a piecewise function. While this approach appears appealing, it is in fact vulnerable to another family of pathological counterexamples. Consider the following function:

\begin{equation}
  \verb_ id(x) = if x != 0 then x else 0_ 
  \label{eq:id}
\end{equation}
This function is semantically identical to the the identity function over $\R$, and as such, its derivative should be equal to 1 at all points. However, the derivative resulting from the construction specified above results in the derivative at \verb_x = 0_ to be erroneously equal to 0. 


There have been forays into formalizing some of the ad-hoc solutions. Huot et al. \cite{huot2020} take the first described approach; they disallow conditionals and implement differentiation as a macro outside of the language syntax. Adabi and Plotkin \cite{abadi2019} introduce non-terminating behaviour at the points of non-smoothness and only guarantee correctness of differentiation at the points that the program terminates. Furthermore, they avoid the problematic program~\ref{eq:id} by imposing a restriction on which predicates can be expressed in their language.

\subsection{Why Distribution Theory?}
\label{sec:why}


In this section we try to give intuition as to why even though our semantics is radically different from existing semantics for differentiable programming, it is still closely related to a couple of methods that practioners currently employ to deal with non-smoothness. In fact, later in section \ref{sec:embed}, we show that we can actually embed a previous language directly into our semantics. 

The existing semantics of differentiable programming that guarantee correctness of differentiation modulo a null-measure set rely on the assumption that, if you want to run gradient descent you may either start with a random initial point or, at every iteration, add a small random noise to the current point. This procedure combined with these correctness theorems allows you to show that with probability $1$ you will always get the correct derivative.

How would you give semantics to these stochastic operations? This is achieved by integrating the almost-smooth function $f$ by the probability measure $\mu$ corresponding to the random noise, which is equal to

$$
\int f \varphi \dint x\text{, where } \varphi \text{ is the probability distribution function of } \mu
$$

This construction looks very similar to the way distribution theory would interpret this quasi-smooth function.

Another way practioners deal with non-smoothness is by finding a smooth approximation to their model. There are a few methods that allow you to define a smooth variant of a model \footnote{See  Pierucci \cite{pierucci2017} for a thorough presentation on some of these methods}. Unfortunately, in general, these methods do not give you a closed form solution. One of these methods relies on the following theorem:

\begin{theorem}
  Let $f : \R^n \to \R$ and $g : \R^n \to \R$ be two measurable functions. If $g$ is smooth and compactly supported then the convolution $(f \star g)(x) = \int_{-\infty}^{\infty} f(y) g(x - y) d y$ is defined everywhere and is smooth as well. 
\end{theorem}

By using the theorem above and choosing an appropriate function $g$ it is possible to find a smooth approximation of $f$. That being said, by closer inspection, if you choose $g$ to be our defined bump function we are, once again, inadvertently getting the same interpretation as the one our semantics would give.

These two examples illustrate that even though our semantics originates from a completely different starting point, it captures some of the methods already being used by practioners.

A more direct consequence of basing our semantics on distribution theory is that this formalism provides a very precise language in which one can model non-smooth behavior. If one were to use the existing semantics to reason about the examples shown in Section~\ref{sec:examples} there would be fundamental flaws in their interpretation. For instance, in the bouncing ball example the derivative of the velocity with respect to time (i.e. acceleration) would be constant everywhere except at the collision point, in which case it would either be undefined or an arbitrary value. Both alternatives do not capture the whole point of the model, which is that the ball is bouncing rather than accelerating at a constant rate. 

A more egregious example is presented by Li et al. \cite{li2018} where it is shown that by considering non-smooth aspects you improve on the state-of-the-art for differentiable rendering techniques. 

For these reasons we believe that distribution theory provides a useful interpretation of differentiable programming.

\section{Distribution Theory for Non-Smooth Modeling}
\label{sec:examples}

In this section we show how distribution theory can be used in non-smooth modeling. As we have already mentioned, one of the main drawbacks of using a semantics which ignores the points of non-smoothness is that it does not properly account for the interaction of differentiation and integration.

Indeed, the popular implementations of integration in automatic differentiation systems use variants of the Monte Carlo method. Therefore, by linearity, the AD procedure distributes the derivative over the various terms of the sum, resulting in unsound behavior. By contrast, a distribution theoretic approach  - like ours - can safely and soundly apply linearity of differentiation. 

\subsection{Coin Flipping Optimization}
\label{sec:example-coin}

In the context of probabilistic programming it is often required to optimize the expected value of a family of random variables. This can be achieved by using gradient-descent methods and, therefore, it requires to compute the derivative of the function

\[
  p \mapsto \int f(p, x) \dint x
\]

Where $f(p, -)$ is a probability density function (pdf). Under certain conditions, derivatives distribute over integrals, reducing the problem of computing the derivative of the integral to computing the derivative of the function $f(-, x)$. 

As a simple example, consider a $p$-biased coin such that we want to optimize its  expected heads frequency. Therefore we have to differentiate the integral $\int_0^1 f_p(x) \dint x$, where $f_p(x) = \mathbb{1}_{0 \leq x \leq p}$.

As we mentioned above, the partial semantics gives $\frac{\partial f_p(x)}{\partial p} = 0$ and, therefore, its integral would be $0$ as well, making it impossible to move away from the initial guess. On the other hand, by considering its distributional derivative, we get $\delta_{p - x}$ and its integral over the interval $[0, 1]$ will be $\mathbb{1}_{0 \leq x \leq 1}$ \cite{bangaru2021}. Once again we see that the partial semantics is inadequate to correctly model non-smooth behavior.

\subsection{Bouncing Ball}

Assume that we drop a ball from a height $h$ with constant acceleration $g$ and that it collides with the ground instantaneously. Its velocity is given by the plot depicted in Figure~\ref{fig:bounce}.


   
   As we can see, there are points of discontinuity at every natural number greater than $0$. This example shows an important difference between existing semantics for differentiable programming and distribution theory. If you consider a language that does not have distribution theoretic primitives, the acceleration of the system is simply $g$, completely ignoring the effects of the collisions. It is possible to show that the distributional derivative of $v(t)$ is $g + \sum_{n=0}^{\infty} \delta_{2n + 1}$, where each $\delta$ captures a time of collision.

   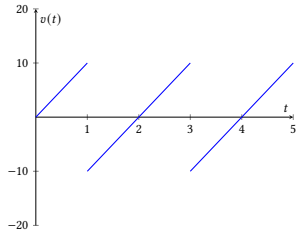
\begin{figure}
     \centering
     \scalebox{.5}{
\begin{tikzpicture}[
  declare function={
    func(\x) = 10*x
    ;
    g(\x) = 10*(x-1) - 10
    ;
    g2(\x) = 10*(x-3) - 10
    ;
  }
]
\begin{axis}[
  axis x line=middle, axis y line=middle,
  ymin=-20, ymax=20, ylabel=$v(t)$,
  xmin=0, xmax=5, xlabel=$t$,
  domain=-10:10,samples=101, 
  ]
  \addplot [domain=0:1,blue,thick] {func(x)};
  \addplot [domain=1:3,blue,thick] {g(x)};
  \addplot [domain=3:5,blue,thick] {g2(x)};
\end{axis}
\end{tikzpicture}
}
\caption{Velocity of a bouncing ball under constant acceleration}
\label{fig:bounce}
\end{figure}

Even though this is a simple physical system, variants of it can be used in differential motion planning algorithms. In such applications it is important to take into consideration the interaction of differentiation and instantaneous collisions \cite{bangaru2021}.







\subsection{Differentiable Ray Tracing}
\label{sec:ray}
In recent years, differentiable methods have found applications in computer graphics. One possible application is to use gradient descent on the rendering algorithm so that you may differentiate the output image with respect to certain scene parameters so that you may, for instance, generate an adversarial picture for an image classifier. Therefore, in order for this method to work the whole image generating process needs to be differentiable, in particular the rendering algorithm. Li et al. \cite{li2018} define a differentiable ray tracing algorithm which fundamentally relies on non-smooth behavior. In this example we will present a simplification of their model. 

Consider a 2D environment parametrized by a continuous space $\Phi$ (e.g. camera position) filled with triangles with different characteristics and with spatial coordinates also captured by $\Phi$. In such a setup the value of a pixel is given by an integral

$$
I = \int \int f(x,y, \Phi) dxdy
$$

In this simple setup the scene function $f$ is simply the sum of functions $f_i$ for each triangle in the scene multiplied by the indicator function of each triangle. In this case, each $f_i$ captures certain properties of the triangles (e.g. its transparency) that are relevant to the rendering of the image and therefore should only be ``active'' within the boundaries of the triangle:

$$
I = \int \int \sum_i \mathbb{1}_{i}f_i(x,y, \Phi) dxdy
$$

Therefore, when trying to reason about the gradient of $I$ we must be careful around the points of non-smoothness in the functions $\mathbb{1}_{i}f_i(x,y, \Phi)$. Following Li et al. \cite{li2018} we have the equation:

\begin{align}
I = \sum_i \int \int \mathbb{1}_{i}f_i(x,y, \Phi) dxdy
\end{align}

We may now take the gradient of $I$:

\begin{align}
  &\nabla \left ( \int \int \mathbb{1}_{i}f_i(x,y, \Phi) dxdy \right) =\\
  &\int \int \nabla \left ( \mathbb{1}_{i}f_i(x,y, \Phi)\right ) dxdy =\\
  &\int\int \mathbb{1}_{i} \nabla f_i(x,y, \Phi) dxdy + \int\int f_i(x,y, \Phi) \nabla \mathbb{1}_{i}  dxdy
\end{align}

What Li et al. have used in their analysis is the fact that the gradient of the indicator function will add a Dirac delta to their integral, which the authors show can be further simplified. This contrasts with the partial semantics, which would simply ignore the non-smoothness and differentiate $f_i(x,y, \Phi)$ without taking into consideration the non-smoothness introduced by the indicator function. Furthermore, this would require commuting the derivative with the integral, which is unsound in this case. Indeed, the authors show that by only using the partial semantics you get a qualitatively worse rendered image. See Section~\ref{sec:validray} for a formal elaboration on how a distribution-theoretic approach resolves this.


\section{Syntax and type system}
\label{sec:syntax}

Figure~\ref{fig:syntax} presents the syntax and type system of $\lambda_\delta$: it is a simply typed $\lambda$-calculus extended with types $\mathcal D(\R^n)$ and $\mathcal D' (\R^n)$ for test functions and distributions, respectively, and some distribution theoretic primitives inhabiting them, which are highlighted. 

We have a differentiation operation, a way of sending a smooth function $f$ to its canonical distribution $T_f$, which we call $\mathsf{lift}$ in our syntax, and a way of defining bump functions
so that we may compute the value of a distribution ``around a point''. Furthermore, we use the fact that distributions can be equipped with a vector space structure and add to our syntax addition and multiplication by a scalar.

To encode conditionals we allow multiplying a smooth function by an indicator function $\mathbb{1}_b$ --- the function such that $\mathbb{1}_b(x) = 1$ if $b(x)$ is true and $0$ otherwise. This construction and the vector space structure allows us to define piecewise smooth functions and define conditionals.

We have also added $\nat$ to the type system and a structural recursion combinator to encode iterators. Intuitively, $\iter t\ u\ n$ composes $u$ with itself $n$ times and is applies this composition to $t$. Note that this iteration procedure is memory-less, i.e. for every step of the iteration the function $u$ does not know in which step it currently is. However, given a function $f : \tau \times \nat \to \tau$, the regular iterator can be used to implement a memoryful one:

$$
\iter \ (t_0, 0) \ (\lambda (t, n).\, (f (t, n), n + 1))
$$

Note that we do not allow to differentiate with respect to a higher-order argument. Instead these features should be used to facilitate the manipulation of functions and distributions.



\begin{example}[Non-smooth programs]

Given two smooth programs $t$ and $u$ and a boolean predicate $b$, we define the program $\mathsf{if }\, b\, \mathsf{ then \,  } t \, \mathsf{ else }\, u$ as the distribution $\mathbb{1}_b t +.  \mathbb{1}_{\neg b}u$
\end{example}

An unexpected consequence of the way we encode if-statements and the typing rule for the indicator function is that we lose the ability to nest conditionals. However, this is not problematic as in the absence of recursive programs, every nested if-statement can be expressed as a single multi-branch if-statement, which is something we can encode as the sum of several indicator functions, each indicator corresponding to an if-statement branch.

That being said, this syntactic restriction is a consequence of a limitation of our model, as morphisms $A \to B$ need to be smooth. However, by Theorem~\ref{th:lift}, every measurable function can be made into a distribution, which suggests that it is theoretically possible to have a more standard presentation of if-statements. For the sake of brevity we will not focus on these issues.

Figure~\ref{fig:typing} presents the typing rules. As it is usually the case, contexts $\Gamma$ are lists of variables and their respective types. We write $\Gamma \vdash t : \tau$ to mean that the program $t$ has type $\tau$ under context $\Gamma$.

\begin{figure}
  \begin{align*}
    \tau := \R\ |\ \R^+ \ |\ \mathsf{Pred}(\R^n)\ |\ \mathbb{N} \ | \tau \times \tau \ |\ \dist{\R^n} \ |\ \test{\R^n} \ |\ \tau \xrightarrow{} \tau
  \end{align*}
  \begin{align*}
    & t, u :=\ x\ |\ r\ |\ (t, u)\ |\ \lettext (x, \ y) = t \text{ in } u \ |\ \\
    & \lettext x = t \text{ in } u\ |\ \lambda x:\tau.t\ |\ t\ u \ | \colorbox{pink} {$\lift (t)$\ |\ $\mathbb{1}_{t}(u)$\ |} \ \\
    & \colorbox{pink}{$ \frac{\partial t}{\partial x_i} $ \ |\  $t\ +.\ u\ $ |\ $t\ ^*.\ u\ $|\ $ \langle t, u \rangle $ | $ \varphi^n(c, r)  \  | \delta_t $}\ | \\ 
    & \iter t\ u\ |\ \text{arithmetic}\ |\ \text{comparators}
  \end{align*}
  \caption{Syntax and type system - distribution theoretic concepts highlighted in red }
  \label{fig:syntax}
\end{figure}

\begin{figure*}
  \centering
$$
\begin{tabular}{c}
  \begin{mathpar}

\inferrule[Lift]{\Gamma \vdash t : \R^n \xrightarrow{} \R}{\Gamma \vdash \lift (t) : \dist {\R^n}}

\and

\inferrule[Indicator Function]{\Gamma \vdash t_1 : \mathsf{Pred}(\R^n) \\ \Gamma \vdash t_2 : \R^n \to \R }{\Gamma \vdash \mathbb{1}_{t_1}(t_2) : \dist{\R^n}}

\and

\inferrule[Differentiation]{\Gamma \vdash t : \dist {\R^n} \\ i \in \{1, ..., n\}}{\Gamma \vdash \frac{\partial t}{\partial x_i}:\dist {\R^n}}

\and

\inferrule[Distribution Application]{\Gamma \vdash t_1 : \dist {\R^n} \\ \Gamma \vdash t_2 : \test{\R^n}}{\Gamma \vdash \langle t_1, t_2 \rangle :\R}

\and

\inferrule[Bump Function]{\Gamma \vdash c : \R^n \\ \Gamma \vdash r : \R^+ \\ n \in \nat}{\Gamma \vdash \varphi^n(c, r):\test{\R^n}}

\and

\inferrule[Dirac delta]{\Gamma \vdash t : \R^n }{\Gamma \vdash \delta_t: \dist {\R^n}}

\end{mathpar}
\end{tabular}
$$
\caption{Selected typing rules}
\label{fig:typing}
\end{figure*}

\begin{example}
  We can write the ReLU function as the following program:

  $$ \cdot \vdash \mathbb{1}_{\lambda x: \R. x \geq 0}(\lambda x:\R. x) : \dist{\R}$$
\end{example}

\begin{example}
  \label{ex:mult}
  We can write a function that transforms a function $f : \R^m \to \R^2$ into a pair of functions of the type $\R^m \to \R$ with the following:
\begin{align*}
      &\cdot \vdash \lambda f. (\lambda x. \text{let }(a, b) = f\ x \text{ in } a,  \lambda x. \text{let }(a, b) = f\ x \text{ in } b)
\end{align*}

\end{example}

By repeatedly applying the above example, we can transform an arbitrary function of type $\R^m \to \R^n$ to an $n$-tuple of functions $\R^m \to \R$.

\begin{example}
We can write a function that takes the derivative of distribution $f : \dist \R$ at some point $x$ and applies to it the bump function centered at $x$ with radius $\varepsilon$ as the following:
    \begin{align*}
        &der : \dist \R \to \R \to \R \to \R\\
        &der\ f\ x\ \varepsilon = \left \langle \frac{\partial f}{\partial x},  \varphi(x, \varepsilon) \right \rangle
    \end{align*}
    where we use Haskell-like syntactic sugar for convenience.
\end{example}

\begin{example}
We can write the gradient descent algorithm running for $n$ iterations, starting at a point $x_0$, for a distribution $f$ and using $\varepsilon$ as the radius of the bump function as the following program:

    \begin{align*}
    &gradDesc : \dist \R \to \R \to \R \to \nat \to \R\\
    &gradDesc\ f\ x_0\ \varepsilon = \\
    &\ \ \ \ \iter x_0 \left ( \lambda x_n:\R^n. x_n - (der\ f\ x\ \varepsilon) \right )
    \end{align*}
    Where again we use Haskell-like syntactic sugar.
\end{example}

In the example above we have defined a gradient descent algorithm for functions $\R\to \R$ which can be used to maximize or minimize the given input function in certain circumstances. 


The situation is a bit subtler when working with functions $\R^n \to \R^m$, since elements of type $\dist{\R^n}$ are, in a way, generalized functions with codomain $\R$, making it unclear at first how to define their gradients using distributions. We get around this by using the universal property of products to factor every function $f : \R^n \to \R^m$ into $m$ functions $f_i : \R^n \to \R$ and defining its gradient using the $m$-tuple of distributional derivatives of $f_i$, as shown in Example~\ref{ex:mult}.


\section{Semantics}
\label{sec:sem}

\begin{figure}
  \begin{center}
  \begin{align*}
   \sem{\R} &= \R \\
   \sem{\R^+} &= \R^+ \\
   \sem{\mathsf{Pred}(\R^n)} &= \mathsf{Pred}(\R^n)\\
   \sem{\nat} &= \nat\\
   \sem{\tau_1 \times \tau_2} &= \sem{\tau_1} \times \sem{\tau_2}\\
   \sem{\test{\R^n}} &= \test{\R^n}\\
   \sem{\dist{\R^n}} &= \dist{\R^n}\\
   \sem{\tau_1 \rightarrow \tau_2} &= \sem{\tau_1} \Rightarrow \sem{\tau_2}
  \end{align*}
\end{center}
  \caption{Type interpretation}
  \label{fig:typesemantics}  
\end{figure}

To every type $\tau$ in our language we associate a diffeological space $\sem \tau$ and every well typed program $\Gamma \vdash t : \tau$ gives rise to a morphism $\sem \Gamma \to \sem \tau$ in $\cat{Diff}$. As usual, we define the semantics by structural induction on the typing derivation. For the sake of clarity, we will write $\sem t$ instead of $\sem {\Gamma \vdash t : \tau}$

\subsection{Interpreting types and non-standard booleans}

  Figure~\ref{fig:typesemantics} shows the semantics of well-typed terms. Since $\cat{Diff}$ is cartesian closed, we interpret the simply typed fragment of our language using the standard constructions. The distribution theoretic types are interpreted as explained in Section~\ref{sec:background}. To interpret the type $\nat$ we will use the fact that $\cat{Diff}$ is cocomplete and therefore can interpret inductive types --- i.e. initial algebras for polynomial functors. Furthermore, we will use the initiality of inductive types to define the semantics of the $\iter$ combinator. Interpreting $\B$ is a bit subtler.\\

  \noindent\textbf{Non-standard booleans.} In cartesian categories with coproducts, booleans are usually defined as the coproduct $1 + 1$, where $1$ is the unit for the cartesian strucuture, and if-statements are defined by the universal property of coproducts. Unfortunately, if we define $\sem \B = 1 + 1$ in $\cat{Diff}$, the only smooth predicates $\R^n \to 1+1$ are the constant functions, which is obviously not expressive enough.

  In order to remedy this we will assume that we have types $\mathsf{Pred}(\R^n)$ whose elements are measurable subsets of $\R^n$. In an earlier version $\lambda_\delta$ we had a type $\B$ of boolean that were equipped with a diffeology such that functions $\R^n \to \B$ were the measurable functions. The problem with this approach is that since we were using the closed structure of $\cat{Diff}$ to interpret predicates, we could partially apply predicates. Unfortunately, currying predicates makes it fairly easy to write non-smooth programs. By using this new approach we can avoid the problems of previous approaches. The key ingredient is equipping $\mathsf{Pred}(\R^n)$ with the constant diffeology.

\begin{theorem}
  $\B = \diff{\mathsf{Pred}(\R^n)}$, where  $\ \mathcal{P}_U^{\mathsf{Pred}(\R^n)}$ are the constant functions $U \to \mathsf{Pred}(\R^n)$, is a diffeological space.
\end{theorem}
\begin{proof}
     This is just a special case of Lemma~\ref{lem:const}. 
\end{proof}

Something appealing about this construction is that it is possible to equip $\mathsf{Pred}(\R^n)$ with a boolean algebra structure.

\begin{theorem}
$\mathsf{Pred}(\R^n)$ is a boolean algebra.
\end{theorem}
\begin{proof}
     This is a direct consequence of Lemma~\ref{lem:constmor}.
\end{proof}

The reason why we care about these non-standard booleans is because they validate the following theorem:

\begin{theorem}
\label{th:indicator}
  The function $\mathbb{1} : \mathsf{Pred}(\R^n) \times (\R^n \to \R) \to \dist{\R^n}$ is a $\cat{Diff}$ morphism.
\end{theorem}

This theorem is used to interpret the $\mathbb{1}$ operator of $\lambda_\delta$ and can be proven analogously to Theorem~\ref{th:smooth}.

\begin{figure}
  \label{fig:sem}
    
    \begin{tabular}{c}
      \begin{mathpar}
        \sem{x}(\gamma) = \gamma(x)

        \and

        \sem{(t_1, t_2)}(\gamma) = (\sem{t_1}(\gamma), \sem{t_2}(\gamma))

        \and

        \sem{\lettext (x, \ y) = t \text{ in } u}(\gamma) = \sem{u}(\gamma, \sem{t}(\gamma))

        \and
        
        \sem{\lambda x : \tau. t}(\gamma) = \lambda x. \sem{t}(\gamma, x)

        \and
        \sem{t\ u}(\gamma) = \sem{t}(\gamma, \sem{u}(\gamma))

        \and
        \sem{\lift (t)}(\gamma) = \mathbf{T}_{\sem t(\gamma)}

        \and
        \sem{\mathbb{1}_t(u)}(\gamma) = T_{\mathbb{1}_{\sem{t}(\gamma)}\sem{u}(\gamma)}

        \and
        
        \sem{\frac{\partial t}{\partial x_i}}(\gamma) = \frac{\partial (\sem{t}(\gamma))}{\partial x_i}

        \and
        \sem{t\ +.\ u}(\gamma) = \sem{t}(\gamma) + \sem{u}(\gamma) 

        \and

        \sem{t\ *.\ u}(\gamma) = \sem{t}(\gamma) * \sem{u}(\gamma) 

        \and

        \sem{\lrang {t,u}}(\gamma) = \lrang {\sem{t}(\gamma), \sem{u}(\gamma)} 

        \and

        \sem{\varphi^n(c, r)}(\gamma) = \varphi_{\sem{r}(\gamma)}^{\sem{c}(\gamma)}

        \and

        \sem{\iter t\ u}(\gamma) = Rec_\nat (\sem{t}(\gamma)) (\sem{u}(\gamma)) 

        \and

        \sem{\delta_t}(\gamma) = \delta_{\sem{t}(\gamma)}

      \end{mathpar}
    \end{tabular}
    \caption{Denotational semantics for $\lambda_\delta$}
    \label{fig:semantics}

  \end{figure}

  \subsection{Semantics of well-typed terms}

  Figure~\ref{fig:semantics} shows the semantics of our constructions. It is mostly standard: the simply typed structure is interpreted by the cartesian closed structure of $\cat{Diff}$ presented in Section~\ref{sec:background}, the iterator is defined as the unique arrow given by the object $\nat$ equipped with the appropriate arrows being an initial algebra\footnote{known as catamorphisms in the functional programming community}. Given an element $x$ of a diffeological space $X$ and an endofunction $f : X \to X$, we name the arrow given by initiality $Rec_\nat \, x \, f$.

  The cartesian closed structure of $\cat{Diff}$ and the fact that the distribution theoretic types are defined as subobjects means that the distribution application is also smooth. The smoothness of the distributional derivative follows by a similar proof. Finally, the validity of the interpretations of $\delta$, $\mathsf{lift}$ and $\mathbb{1}$ are given by Theorem~\ref{th:dirac}, Theorem~\ref{th:lift} and Theorem~\ref{th:indicator}.

  Note that the construction of our non-standard booleans impose the predicates in our language to be measurable sets. In practice, however, it is extremely hard to define a non-measurable set, making this restriction almost non-existent.

To reiterate, even though $\cat{Diff}$ might look a bit too abstract at times, morphisms $\R^n \to \R^m$ are exactly the smooth functions between $\R^n$ and $R^m$. Furthermore, since $\dist{\R^n}$ is also exactly equal to the set of distributions over $\R^n$, the vast catalog of theorems from distribution theory can be used to reason about programs.

\subsection{Well-behaved conditionals}

As we illustrate in Appendix~\ref{sec:impl} with PyTorch's max function, the current handling of the differentiation of non-smooth programs and conditionals in production environments can lead to unexpected aberrant behavior. The key problem that gives rise to this behavior is that Theorem~\ref{th:intro} does not hold in existing semantics in the context of differentiation:

Our semantics validates this equation  with the following theorem:

\begin{theorem}
  \label{th:main}
  For every context $\Gamma$ and well-typed programs $\Gamma \vdash b : \mathsf{Pred}(\R^n)$, $\Gamma \vdash t, u, e : \R^n \to \R$, if $\sem{b}(r) = ff$ implies $\sem{t}(r) = \sem{u}(r)$ and $\sem{b}(r) = tt$ implies $\sem{t}(r) = \sem{e}(r)$then we have the equality $\sem{\mathbb{1}_b e +. \mathbb{1}_{\neg b}u} = \sem {\lift{(t)}}$.
\end{theorem}
\begin{proof}
     $\sem{\mathbb{1}_b e +. \mathbb{1}_{\neg b}u} = 
     \sem{\mathbb{1}_b t +. \mathbb{1}_{\neg b}t} = \sem {\lift{(t)}}$
\end{proof}

This theorem implies Theorem~\ref{th:intro}, resolving this issue.


\section{Case studies}
\label{ref:prac}

\subsection{Bouncing Ball}

Let us consider a variant of the example presented in Section~\ref{sec:examples} where there is a ball moving at constant velocity $v$ perpendicular to a wall and it elastically hits a wall at time $t_{wall}$ and moves in the opposite direction with velocity $-v$. The program that computes the velocity is given by:

\begin{equation}
    u = \mathbb{1}_{t <= t_{wall}}(\lambda t.\, v) +. \mathbb{1}_{t > t_{wall}}(\lambda t.\, -v)
    \label{eq:ball}
\end{equation}

By unfolding the definitions we can easily show $\sem{\frac{\partial u}{\partial t}} = \delta_{t_{wall}}$. As it is usually the case, the Dirac delta is modelling the time of collision with the wall.

Existing semantics of differentiable programming fail in two ways when trying to model phenomena with non-smooth behavior like this. The first is to make the whole phenomenon inexpressible by disallowing any form of non-smooth behavior. This means that \ref{eq:ball} and other simple physical phenomena are completely inexpressible in these semantics. Nonetheless, this approach was used by \cite{huot2020} and \cite{ehrhard2003} among others.

The second approach is somewhat more interesting, in that it utilizes a form of partial semantics that completely ignores the point of collision, such as the semantics of \cite{abadi2019}. In their semantics, it is possible to express a form of \ref{eq:ball} as the following program:

\[
M = \ifthen{t < 0}{v}{-v}
\]

However, by differentiating the program above using their semantics we get 

\[\sem{\frac{\dint M}{\dint t}} = 
\begin{cases}
0\ \ & \text{if } t \neq 0\\
\bot\ \ & \text{otherwise}
\end{cases}\]

Were we to attempt to model the position of the ball using the velocity with these existing semantics, the model would not be able to distinguish between the simple physical reality of the ball bouncing off of the wall and the ball passing through it. More recently, Matthijs Vakar has also posted some preliminary work on extending this flavor of partiality semantics on the arXiv, but their approach still fails to model this phenomena. Even the approach in \cite{sherman2020} with its advances in modeling non-smooth functions is unable to distinguish between these situations as they also force the function into undefined behavior at points of discontinuity.

\subsection{Derivatives of Intergrals}
In this example we will show how our syntax and semantics can properly deal with the interaction of differentiation and integration. We illustrate this by coming back to the example in Section~\ref{sec:example-coin} showing that the equation $ \frac{\dint }{\dint p} \int \mathbb{1}_{0 \leq x \leq p} \dint x = \mathbb{1}_{0 \leq x \leq 1}$ holds in our semantics.

First, we present a method of computing integrals of compactly supported functions using our semantics. Let $K_1 \subset K_2$ be two compact subsets of $\R^n$ such that $K_1$ is strictly contained in $K_2$. There exists a smooth function $\psi$ such that $0 \leq \psi \leq 1$, $\psi|_{K_1} \equiv 1$ and $\psi|_{\R^n \backslash K_2} \equiv 0$ \cite{lee2013}. Because this function is smooth and has compact support, it allows us to effectively construct a test function that is constant on any arbitrary compact set. Intuitively, this function is $1$ on $K_1$, transitions smoothly on from $1$ to $0$ on $K_2 \setminus K_1$, and is $0$ everywhere else.

We will denote this function as $\psi_{K_1}^{K_2}$. If we apply a distribution $T_f$ to it we get:
\begin{align*}
  &\lrang{T_f, \psi_{K_1}^{K_2}} = \int_{\R^n} f(x) \psi_{K_1}^{K_2}(x) \mathop{}\! \mathrm{d} x = \\
  &\int_{K_1} f(x) \mathop{}\! \mathrm{d} x + \int_{K_2 \backslash K_1} f(x) \varphi_{K_1}^{K_2}(x) \mathop{}\! \mathrm{d} x
\end{align*}
    
Importantly, if $f$ has compact support, then by carefully choosing $K_1$ such that $\text{supp}(f) \subseteq K_1$, we can get the equality $\lrang{T_f, \psi_{K_1}^{K_2}} = \int_{K_1} f(x) \dint x$.

We want to compute the derivative of $\lambda p. \int \mathbb{1}_{0 < x < p}$, which, as we explained, can be computed by the term

\[
  t = \frac{\partial}{\partial t}(\lift (\lambda\, t.\, \lrang{\mathbb{1}_{(\lambda \, x. \, 0 < x < t)}, \psi_{[ 0, 1 ]}^{K_2}}))
\].

Now it is easy to show that $\sem{t} = \sem{\mathbb{1}_{0 \leq x \leq 1} (\lambda x. 1)}$

    Note that while we do not explicitly have the $\psi_{K_1}^{K_2}$ terms in our syntax, if we add constructors $\psi_{K_1}^{K_2}$ to $\lambda_\delta$ where $K_1$ and $K_2$ are cubes or spheres such that $K_1 \subset K_2$ then we would already be able to program many interesting examples, as we will see in the next section.
    
    As many applications heavily rely on the interaction of integrals and differentiation\cite{bangaru2021,li2018,bangaru2020}, it is important to develop semantics that can soundly justify these equations. Indeed, this kind of reasoning is frequently used by practioners, as illustrated by Bangaru and Michel et al. \cite{bangaru2021}, where these ideas have been used to automatically rewrite programs in the differentiable programming language TEG. It is important to note that their system cannot properly reason about the example above.

While the programming model imposed by working with distributions cannot compute the value of a function at a point, we can approximate it by applying a bump function with a sufficiently small radius. Unfortunately this becomes problematic when considering what an operational semantics to such language would look like, as applying a test function to a distribution that originated from a smooth function is integrating their product, meaning that an operational semantics for $\lambda_\delta$ would require computing integrals.

For implementation purposes, however, the point above is not too problematic, as there are many algorithms that can efficiently compute approximations to arbitrary integrals --- we have implemented a small library to demonstrate the use of distributions and key concepts from our language in a software artifact. See Appendix~\ref{sec:impl} for more details.



\subsection{Validating Equations for a Ray Tracing Algorithm}
\label{sec:validray}

We have already argued in Section~\ref{sec:ray} that using a non-distribution-theoretic semantics makes it impossible to validate the equations used by Li et al. In this section we are going to show how we can use $\lambda_\delta$ and its semantics to correctly reason about a simplified implementation of the ray tracing algorithm.

As we have shown above, $\lambda_\delta$ provides a simple way of computing the integrals of compactly supported functions.

\begin{align*}
  &integral : \dist {\R^2} \to \R^2 \to \R^2 \to \R \\
  &integral \ T \ (x_1, y_1) \ (x_2, y_2) \ = \lrang{T, \psi_{K_1}^{K_2}},
\end{align*}

Where $K_1 = [x_1,x_2] \times [y_1, y_2]$ and $K_2 = [2x_1, 2x_2] \times [2y_1, 2y_2]$.
Next, we need to write a program that computes the integrand $f(x, y, \Phi)$. For the sake of convenience, we extend $\lambda_\delta$ with lists which is semantically valid because $\cat{Diff}$ is cocomplete and, therefore, can interpret inductive types and their inductive principles, which for lists is the familiar fold. Syntactically $\lambda_\delta$ will look exactly like a regular $\lambda$-calculus equipped with inductive types; see Huot et al. \cite{huot2020} for more details.

With this extension we can write a program that, given a list of triangles and their characteristic functions, returns the sum of distributions $\sum_i \mathbb{1}_{\alpha_i}f_i$, where $\alpha_i$ is the predicate for the $i$-th triangle and $f_i$ is its characteristic function:

    \begin{align*}
        &charFunc : \Phi \to [(\mathsf{Pred}(\R^n), \R^2 \times \Phi \to \R)] \to \dist{\R^2}\\
        &charFunc\ \underscore \ [ \, ] = 0\\
        &charFunc\ \varphi\ ((f_i, \alpha_i)::tl)  = \mathbb{1}_{\alpha_i}(\lambda r.\ f_i(r, \varphi)) +_. (charFunc \ \varphi \ tl)
    \end{align*}

    Once again, note that the program above is structuraly recursive and thus can be implemented with a fold without having to make use of full recursion. We can now implement the pixel value function $I$:

    \begin{align*}
        &I : [(\mathsf{Pred}(\R^2), \R^2 \times \Phi \to \R)] \to \R^2 \to \R^2 \to (\Phi \to \R) \\
        &I\ l \ (x_1, x_2)\ (y_1, y_2) = \lambda \varphi.\ integral \ (charFunc \ \varphi \ l) \ (x_1, x_2) \ (y_1, y_2)
    \end{align*}


    Assuming that the support of $charFunc \ \varphi \ l$ is a subset of the integral domain, it is easy to show that the semantics of I is exactly the formula presented by Li et al. By unfolding the semantics of $\nabla (\lift \ (I\ l\ (x_1, x_2)\ (y_1, y_2)))$ we can see where adopting a non-distributional semantics would be problematic:


    \begin{align*}
      &\sem{\nabla (\lift \ (I\ l\ (x_1, x_2)\ (y_1, y_2))} = \nabla \sem{(\lift \ I\ l\ (x_1, x_2)\ (y_1, y_2))} = \\
      &\nabla \int_{(x_1,y_1)}^{(x_2, y_2)} (charFunc \ \varphi \ l) = \nabla \int_{(x_1,y_1)}^{(x_2, y_2)} \sum_i \mathbb{1}_{\alpha_i}(\lambda r.\ f_i (r, \varphi)) =\\
      & \sum_i \nabla \int_{(x_1,y_1)}^{(x_2, y_2)} \mathbb{1}_{\alpha_i}(\lambda r.\ f_i (r, \varphi))
    \end{align*}

    Existing AD algorithms will always commute with integrals --- as they are implemented using finite sums --- even when the Leibniz theorem \footnote{\url{https://en.wikipedia.org/wiki/Leibniz_integral_rule}} does not hold. Fortunately, distributions do not suffer from this drawback, so our semantics would be able to commute $\nabla$ and $\int$ and soundly apply the last steps of the reasoning done in Section~\ref{sec:ray}. 
    
    The incompatibility of regular AD and integration has been observed by Bangaru, Michel et al in recent work\cite{bangaru2021}, where they have coined the terms "discretize-then-differentiate" and "differentiate-then-discretize" to contrast the standard approach to AD with the distribution-theoretic one.

\section{Translating to and from other differentiable languages}

\subsection{Embedding a Smooth $\lambda$-calculus}
\label{sec:embed}

It is an important question to understand how our language and semantics relate to existing semantics of differentiable programming. In this section we will show how the language proposed by \cite{huot2020} can be soundly embedded in our language and how their AD program transformation relates to our differentiation operation.

In their work they define a simply typed $\lambda$-calculus with one base type for the real numbers and smooth primitives (e.g. the sine function). Their differentiation program transformation follows the dual number approach to AD, where each input carries an extra parameter corresponding to the derivative in that argument. This is achieved by defining the following type transformation:

\begin{align*}
    \D(\R) &= \R \times \R\\
    \D(\tau \times \tau) &= \D(\tau) \times \D(\tau)\\
    \D(\tau \to \tau) &= \D(\tau) \to \D(\tau)
\end{align*}

They also define a transformation $\D(-)$ at the term level and prove that if $\Gamma \vdash t : \tau$ then $\D(\Gamma) \vdash \D(t) : \D(\tau)$. Their translation is elegant and can be shown to be functorial, which is a consequence of its compositionality. Their (simplified) correctness property is:

\begin{Lemma}[\cite{huot2020}]
\label{lem:huot}
If $x : \R \vdash t : \R$ then for every smooth function $f : \R \to \R$, $(f, f' ); \sem{\D(t)} = (f; \sem{t}, (f; \sem{t})')$, where $(-)'$ is the derivative operation.
\end{Lemma}

Their full correctness theorem takes into account arbitrary open term such that the inputs and outputs are smooth manifolds. For sake of presentation, we focus on this simplified case. Due to their language also being a $\lambda$-calculus that is interpreted $\diffcat$, the identity translation is well-typed and so is their AD translation. Now, we can state the theorem:

\begin{Th}
\label{th:statonsound}
 Let $x : \R \vdash t : \R$ be a well-typed program in the original language. Consider the terms $\cdot \vdash t_1 = \lamb x {\pi_1 (\D(t)(x,1))} : \R \to \R$ and $t_2 = \lamb x {\pi_2 (\D(t)(x,1))} : \R \to \R$, then
 
 \[
    \sem{(\lift{(\lamb x t)}, \frac{\partial}{\partial x}(\lift (\lamb x t)))} = \sem{(\lift (t_1), \lift (t_2))}
 \]
\end{Th}
\begin{proof}
 The proof is a straightforward application of Theorem~\ref{th:smoothdist} and Lemma~\ref{lem:huot} with $f = \lamb x x$.
\end{proof}


    

By using the general soundness theorem of \cite{huot2020}, slightly massaging the statement above and modifying the programs $t_1$ and $t_2$ it is possible to prove a similar version of the theorem above for open terms $x_1 : \R, \cdots , x_n :\R \vdash t : \R^m$.

\subsection{Constructive Semantics}
\label{sec:construct}

An important aspect of the semantics presented by \cite{sherman2019} is that their language has datatypes both open and compact subsets of topological spaces, which in their language are type constructors $\texttt{OShape}$ and $\texttt{KShape}$. Due to idiosyncrasies of constructive topology, they also make heavy use of the type $\texttt{OShape}\, E \times \texttt{KShape}\, E$, which they call $\texttt{OKShape}$. In the original papers \cite{sherman2019,sherman2020} they go over the formalities that make their semantics work and be computable. In this section we are only interested in their provided API to program with these spaces which, in particular, makes it possible to compute integrals over compact domains of $\R^n$. This suggests that it should be possible to translate $\lambda_\delta$ into their metalanguage. We define such a translation $\bsem{-}_\delta$, starting with the type translation depicted in Figure~\ref{fig:typetrans}.

\begin{figure}
  \begin{center}
  \begin{align*}
    \bsem{\R}_{\delta} &= \mathfrak{R}\\
    \bsem{\tau_1 \times \tau_2}_{\delta} &= \bsem{\tau_1}_\delta \times \bsem{\tau_2}_\delta\\
    \bsem{\tau_1 \to \tau_2}_{\delta} &= \bsem{\tau_1}_\delta \to \bsem{\tau_2}_\delta\\
    \bsem{\mathsf{Pred}(\R^n)}_{\delta} &= \texttt{OKShape} (\mathfrak{R}^n)\\
    \bsem{\test{\R^n}}_{\delta} &= (\mathfrak{R}^n \to \mathfrak{R}) \times (\texttt{OKShape}\,  \mathfrak{R}^n)\\
    \bsem{\dist{\R^n}}_{\delta} &=  \bsem{\test{\R^n}}_\delta \to \mathfrak{R}
\end{align*}
\end{center}
  \caption{Type translation into $\lambda_s$}
  \label{fig:typetrans}  
\end{figure}

Since their semantics is based on constructive topology, they only have a datatype for the constructive real numbers $\mathfrak{R}$. The type constructors of the simply-typed $\lambda$ calculus are standard. The interesting aspects are the distribution theoretic primitives. We interpret predicates as open sets corresponding to their indicator functions. Test functions are interpreted as pairs of a function a compact set, i.e. its support. Then, as it is standard, distributions are functions from test functions to real numbers. 

\begin{figure}
  \begin{center}
  \begin{align*}
\bsem{\varphi_r^c}_\delta &= (\varphi_r^c,\, \app{\texttt{makecube_n}}{r}\, c)\\
\bsem{\lift \, t}_\delta &= \lamb {(\varphi, K)} \int_K\lamb{x}{(\app{\varphi}{x})(\app{\bsem{t}_\delta}{x})}\\
\bsem{\mathbb{1}_b\, t}_\delta &= \lamb {(\varphi, K)} {\letin {K'}  {( \bsem{b}_\delta \cap K)} {\int_{K'} \, (\lamb x {\bsem{t}_\delta(x) * (\app{\texttt{indicator}}{K'}\, x)\, })}}\\
\bsem{\delta_c}_\delta &= \lamb{(\varphi, K)}{\app{\varphi}{c}}\\
\bsem{\frac{\partial t}{\partial x_i}}_\delta &= \lamb{(\varphi, K)}{- \app{\bsem{t}_\delta}{(\app{(\texttt{derivative}_i}{\varphi}), K)}}
\end{align*}
\end{center}
  \caption{Selected term translations into $\lambda_s$}
  \label{fig:termtrans}  
\end{figure}

We present parts of the term translation in Figure~\ref{fig:termtrans} and, once again, the translation for the lambda calculus syntax is trivial. The interesting aspects are some of the distribution theoretic primitives. The primitive test functions $\varphi_r^c$ are mapped to their mathematical counterparts $\varphi_r^c$, which are simply a multiplication of exponentials, and to the hypercube centered around $c$ and with sides of length $r$, which can be defined using the primitives\footnote{
\url{https://github.com/psg-mit/marshall/blob/master/examples/stoneworks/krep.asd}\\
\url{https://github.com/psg-mit/marshall/blob/master/examples/stoneworks/orep.asd}\\
\url{https://github.com/psg-mit/marshall/blob/master/examples/stoneworks/okrep.asd}}
\begin{align*}
&\texttt{unit_cube } : \texttt{KShape}\,\mathfrak{R}\\    
&\texttt{product} : \texttt{KShape}\, \mathfrak{R^m} \to \texttt{KShape}\, \mathfrak{R^m} \to \texttt{KShape}\, \mathfrak{R^{m + n}}\\
&\texttt{translate} : \mathfrak{R}^n \to \texttt{KShape}\, \mathfrak{R}^n \to \texttt{KShape}\, \mathfrak{R}^n\\
&\texttt{scale} : \mathfrak{R} \to \texttt{KShape}\, \mathfrak{R}^n \to \texttt{KShape}\, \mathfrak{R}^n\\
\end{align*}
These primitives are also available for $\texttt{OShape}$, making it possible to define an $\texttt{OKShape}\, \mathfrak{R}^n$ for $n$-dimensional hypercubes.

The lift primitive uses the iterated integral primitive:
\[
\mathsf{integral_n}: \texttt{OKShape}\, \mathfrak{R}^n \to (\mathfrak{R}^n \to \mathfrak{R}) \to \mathfrak{R} 
\]

This is simply an iterated application of their primitive 
\[\mathsf{integral} : \texttt{OKShape}\, \mathfrak{R} \to (\mathfrak{R} \to \mathfrak{R}) \to \mathfrak{R}
\]
for one dimensional integration over compact support in order to support higher dimensional integration. We use the syntactic sugar $\int \triangleq \texttt{integral}_n$. 

In order to translate the indicator function we use their primitive 
\[\cap : (OKShape\, \mathfrak{R}^n) \to (OShape\, \mathfrak{R}^n) \to (OKShape\, \mathfrak{R}^n)\] 

which computes the intersection of a compact subset with an open subset, resulting in a compact subset. Then, we simply compute the integral over this intersection of the function $f(x)*\app {(\mathsf{indicator}} {(b \cap K)} \, x) $, where the function $\mathsf{indicator} : OShape\, \mathfrak{R}^n \to \mathfrak{R}^n \to \mathfrak{R}$ returns $1$ when $x$ is in the interior of the shape, and $0$ when outside its boundary. Note that for points at the boundary, this function returns an indeterminate value.

Dirac deltas are translated the standard way. In order to translate derivatives we make use of the primitive:
\[
\texttt{derivative}: (\mathfrak{R} \to \mathfrak{R}) \to (\mathfrak{R} \to \mathfrak{R})
\]

which, by making use of partial application of the test function, allows us to define partial derivatives of functions $\mathfrak{R}^n \to \mathfrak{R}$ as the authors show in their paper \cite{sherman2020}, which we denote by the syntactic sugar $\texttt{derivative}_i$. Just an example, a partial derivative of a function $f : \mathfrak{R}^2 \to \mathfrak{R}$ can be defined as $\lamb x {\app{\texttt{derivative}} {(\lamb y {f\, x\, y}})}$.
The other distributive theoretic primitives such as distribution application and the vector space structure use the standard translation.
\begin{Th}
  If $\Gamma \vdash_{\lambda_\delta} t : \tau$ then $\bsem{\Gamma}_\delta \vdash_{\lambda_s} \bsem{t}_\delta : \bsem\tau_\delta$.
\end{Th}
\begin{proof}
  The proof follows by induction on the typing derivation $\Gamma \vdash_{\lambda_\delta} t : \tau$, using the type judgements of the API functions defined above. 
\end{proof}

We could also prove a similar theorem here as we did in \ref{th:statonsound}, where the proof would be almost identical, save for the fact that soundness of differentiation is given by construction of the primitive here rather than as a meta-property of a defined macro. 

This illustrates that there are other semantic domains that can soundly give semantics to $\lambda_\delta$. This exercise should not be interpreted as our language not extending the semantics defined by Sherman et al. Instead, we see this as a valuable addition to their semantics, since previous work has shown how ignoring jump discontinuities -- as it is the case in their base semantics -- ignores important physical interactions in your model, as demonstrated in \cite{li2018,bangaru2021}. Furthermore, since this semantics is fully computable and numerically stable, it sidesteps the undesirable property of noncomputable integrals in the $\cat{Diff}$ semantics. 

\section{Implementation}
\label{sec:impl}

We have implemented a proof-of-concept library for the main concepts in $\lambda_\delta$ in PyTorch\footnote{\url{https://github.com/pytorch/pytorch}}. The distributional derivative is defined in terms of regular derivatives, which opens two possibilities for implementation: we either define an automatic differentiation procedure suited to distributions, or we use an out-of-the-box AD procedure and apply it to the test functions. There are advantages and disadvantages to both approaches. By construction, the distributions definable in $\lambda_\delta$ have an easy to characterize normal-form. Most of the syntactic constructions have an easy to define interaction with derivatives --- e.g. the derivative of a lift is the lift of the derivative --- with the exception of the multiplication of a smooth function by a predicate. The feasibility of directly differentiating a predicate is conditional on its complexity. One way around this problem is by only allowing simple predicates that have well-known derivatives. For example, the derivative of the predicate $\lambda x.\ x \leq 0$ is $\delta_0$. Unfortunately when dealing with higher-dimensional objects it becomes harder to compute these derivatives. The advantage of this approach is that it is closer to heart to an important selling point of AD methods: it is possible to share computations between the computation of a function and its derivatives.

While the test function approach does not share this computation sharing property, it has the advantage of being incredibly easy to implement. They are by definition infinitely differentiable and are simple enough to easily apply any existing AD method to implement their differentiation. It is for this reason that we take the test function approach in our library.

Another key aspect of our library is distribution application as it relies on higher-dimensional integration. This feature goes out of the scope of PyTorch's standard library, so we have used torchquad's\footnote{\url{https://github.com/esa/torchquad}} Monte Carlo integration implementation. Because we are using Monte Carlo integration methods, our distribution application is not strictly deterministic and has some variance from application to application. This instability further compounds at higher derivatives, though it can be tamed to a degree by sampling a larger volume of points for the Monte Carlo method. However, because this implementation is meant to be more of a toy implementation of a few key features of our language rather than an optimized compiler, we have not attempted to optimize for further performance. 

Our implementation resolves some of the strange behaviour that arises from differentiating conditionals (such as the maximum function) in PyTorch. For instance, consider the following three functions:

\begin{verbatim}
    def sillyID_1(x):
        return max(0, x) - max(-x, 0)
        
    def sillyID_2(x):
        return max(x, 0) - max(-x, 0)
        
    def sillyID_3(x):
        return max(0, x) - max(0, -x)
\end{verbatim}

These are each a different way of implementing the identity function which appear to be equivalent save for the order of the arguments. However, each of these functions will return a different derivative at the origin when fed into PyTorch's AD implementation. More specifically, the first will return 1, the second will return 0, and the third will return 2. Our implementation resolves this inconsistency modulo the noise from numerical integration. 

Finally, if we were to implement an actual compiler to our language, we would avoid using integration as much as possible, as it is inefficient and susceptible to floating point errors. In an actual optimizing compiler for $\lambda_\delta$ we could rewrite distribution applications $\lrang {\mathsf{lift}\ t,\varphi^n(c, r)}$ as $t\, c$, which would save computation when evaluating away from boundaries of discontinuity. We could apply a similar procedure for differentiation and apply normal AD algorithms far from these boundaries. 

\section{Related Work}
\renewcommand{\arraystretch}{1.5}
\def\checkmark{\tikz\fill[scale=0.4](0,.35) -- (.25,0) -- (1,.7) -- (.25,.15) -- cycle;} 
\newcommand{\cmark}{\ding{51}}
\newcommand{\xmark}{\ding{55}}

\begin{center}
    \begin{tabular}{ | p{1.9cm} | p{2cm} | p{1.8cm} | p{2.0cm} | l |}
    \hline
     & Higher-order Functions & Higher Derivatives & Non-Smooth Conditionals & Theorem~\ref{th:main}\\ \hline
    \cite{abadi2019} & \xmark & \cmark & \cmark(Partiality) & \xmark \\ \hline
    \cite{huot2020} & \cmark & \cmark & \xmark & N/A \\ \hline
    \cite{sherman2020} & \cmark & \cmark & \xmark (Locally\newline Lipschitz) & N/A \\ \hline
    \cite{bangaru2021} & \xmark & \cmark & \cmark & \cmark \\ \hline
    \cite{ehrhard2003} & \cmark & \cmark & \xmark & N/A \\ \hline
    This Work & \cmark & \cmark & \cmark & \cmark \\ \hline

    \end{tabular}
\end{center}


\paragraph{Recursive languages} 

Abadi and Plotkin \cite{abadi2019} define a first-order programming language with a reverse-mode AD construct and while loops. They define the semantics of their language using the fact that the set of infinitely differentiable partial functions $\R^n \rightharpoondown \R^m$ forms a pointed CPO. They define an operational semantics that implements an AD algorithm and prove it adequate with respect to their denotational semantics. Their language does not support higher-order functions and uses non-termination to deal with jump discontinuities. In order to prove the correctness of automatic differentiation everywhere the program terminates they work with partial predicates $p : R^n \rightharpoondown \{ tt, ff\}$ such that both $p^{-1}(\{ tt \})$ and $p^{-1}(\{ ff \})$ are open sets. However, adding non-terminating behavior at points of discontinuity loses the expressive power required to describe models such as those in \ref{sec:examples}.


\paragraph{Semantics for Macro-Based AD}
Huot et al. \cite{huot2020} also use the category $\cat{Diff}$ to give semantics to a differentiable $\lambda$-calculus. They define a global program transformation corresponding to an implementation of forward-mode automatic differentiation which they show corresponds to the semantic differentiation by a logical relations argument. While they use the same semantic model we do, their language lacks the distribution theoretic machinery developed here. As a result, their language lacks a conditional construct beyond pattern matching on tuples, which severely limits the programs expressible. In fact, none of the examples presented in Section~\ref{sec:examples} are expressible in their language. Furthermore, as we have shown in Section~\ref{sec:embed}, our semantics can be seen as a conservative extension of their semantics.

Matthijs Vakar has a long line of work on defining languages where an automatic differentiation operator is sound. This began with his collaboration with Huot and Staton \cite{huot2020}, and has consistently used a smooth model similar to the one described therein with the same limitations on conditional statements. However, in \cite{vakar2020} posted on arXiv (and thus we only tenatively include it given its lack of peer review), he presents a model of a differentiable programming language with a $\textbf{sign}$ construct, which gives rise to non-smooth behavior in programs. He uses an approach similar to \cite{abadi2019} in that he leaves $\textbf{sign}$ a partial function that is undefined at zero, which gives rise to the same expressiveness issues. It is unclear at the moment if there is a sound translation from his language into ours, since our programs are total.
 
\paragraph{Constructive Semantics}
While most of the existing approaches (described in \ref{sec:ifstmt}) cannot handle differentiating directly at points of discontinuity, the semantics defined by Sherman et al. \cite{sherman2020} uses ideas from constructive topology to interpret a differentiable language that admits higher-order functions and locally-Lipschitz functions that may be differentiated. What distinguishes their approach from ours is that to interpret non-smooth programs they use the idea of subgradient, formally expressed by the concept of \emph{Clarke derivatives}. A consequence of using Clarke derivatives is that the higher derivatives of their non-differentiable locally-Lipschitz functions are undefined. 

To interpret the higher-order fragment they use a sheaf-theoretic construction which is similar to the one used for $\cat{Diff}$. They define a tangent bundle functor using Kan extensions, similar to the construction presented by Staton et al. \cite{staton2016} for the Giry monad. A consequence of their semantics is that even though their semantic category has coproducts and therefore, admits pattern matching, the only predicates $\R^n \to 2$ available are the constant ones, because morphisms are continuous and $\R^n$ is connected but $2$ is not. This prohibits expressing discontinuous functions such as the Heaviside function in their language. Furthermore, even though they would satisfy a kind of Theorem~\ref{th:main}, it would be trivially true, since the predicate has to be constant. In short, while their language uses impressive mathematical machinery to constructively express a much larger class of functions than previous differentiable languages, local Lipschitz continuity is not sufficient for many useful conditionals, and Clarke derivatives fail to capture the true behavior of even those locally-Lipschitz continuous functions at higher derivatives. 

In Section~\ref{sec:construct} we have shown how it is possible to embed $\lambda_\delta$ in their language. We also see this as our language subsuming their original semantics, since non distribution theoretic semantics for nondifferentiable programs will inevitably ignore jump discontinuities. 

\paragraph{Differential Linear Logic}
Since the turn of the century, kickstarted by Ehrhard and Regnier \cite{ehrhard2003}, much work has been done in bridging the gap between differentiation and logic. This has led to the discovery of differential linear logic (DiLL), the differential $\lambda$-calculus and the categorical formulation of differentiation. Though much work has been done on these categorical formalisms, it is still not clear how these models could handle certain features that are expected from the programming languages community -- recursion and if-statements being two of those. That being said, some models of the differential $\lambda$-calculus are connected to distribution theory. Kerjean and Tasson \cite{kerjean2018mackey} have defined a model of differential linear logic where the exponential $!A$ is interpreted as the compactly-supported distributions over $A$. Further research is needed to understand if there is a treatment of DiLL that can handle non-compactly-supported distributions.


\paragraph{Distribution Theory in Computer Science}
The idea of using distribution theory to interpret jump discontinuities has also been used in \cite{nilsson2003} to write a Haskell library for functional reactive programming. However, their approach required two restrictive preconditions: the program had to be provided with the locations of the discontinuities, and the number of discontinuities had to be finite. 



Finally, a distribution-theoretic semantics is at the core of the differentiable language TEG \cite{bangaru2021} --- once again showing that ideas from distribution theory are already used by practioners of differentiable programming. They have defined an untyped, first-order language that has both differentials as well as integral primitives. With their distribution theoretic semantics they focus on reasoning equationally about the differentiation of the integral of non-smooth functions. The main drawbacks of their language when compared to $\lambda_\delta$ is that the fact that they are untyped create some restrictions on the programs they can write. Besides, since $\lambda_\delta$ has higher-order functions it provides more expressive primitives to the programmer. It would be interesting future work to extend the TEG language with a type system and higher-order functions so that $\lambda_\delta$ could be seen as its idealized core calculus.
\section{Future work and Conclusion}

We have defined a denotational semantics for $\lambda_\delta$, a higher-order differentiable language extended with  distribution theoretic primitives which provides a solution to the if-statement problem in differentiable programming. Our semantics is the first that validates certain expected if-statements equations. We highlight the fact that there might other interesting models to the calculus presented here. 

For future work we would like to better understand how our semantics could be used to study the solution of linear partial differential equations. Something similar was done by Kerjean \cite{kerjean2018logical} using compactly supported distributions. By adding recursion to $\lambda_\delta$ and defining its denotational semantics we would have syntax to express the solution of differential equations using the fixed-point operator of our language, making the connections to physics simulators even more explicit. Furthermore, we conjecture that by allowing non-terminating behavior it might be easier to get a better categorical understanding of what distributions are, as test functions might simply be smooth functions, instead of compactly supported ones.

Another promising line of work that requires further research is developing a theory of probability inside $\cat{Diff}$. Many modern Bayesian inference engines rely on differentiable programming. Therefore, it is paramount to develop a theory that encompasses both differential and probabilistic primitives.

\bibliographystyle{ACM-Reference-Format}
\bibliography{biblio}
  

\appendix
\section{Distribution theory}
\label{sec:background}

The original motivation of distribution theory was studying the solutions of linear partial differential equations. Many constructions for functions with codomain $\R$ have a distribution theoretic analogue, which is why they are usually referred to as ``generalized functions''. In our case we are interested in the theory of differentiability of distributions, as they allow us to differentiate certain functions with jump discontinuities or other points of non-differentiability. Since distribution theory is not commonly used by researchers in programming languages this section serves as a self-contained introduction to the subject --- see Golse \cite{golse2010} for a more detailed presentation.

\subsection{Definitions}

\begin{Def}
  Let $X$ be a topological space and $f : X \to \R$ be a function. We define the \emph{support of $f$} as $\mathrm{supp} (f) = \overline{\set {x} { f(x) \neq 0}}$, where $\overline{A}$ is the topological closure of a subset $A$.
\end{Def}

Loosely speaking, the support of a real-valued function is the set of points where the function is non-zero. We say that the support is \emph{compact} if it is bounded. Throughout this section $U \subseteq \R^n$ will always be assumed to be an open set.

\begin{Def}
  $C^\infty_{c}(U)$ is the set of compactly supported infinitely differentiable functions $f : U \to \R$, we may also call this set $\mathcal D(U)$. The elements of this set are usually referred to as \emph{test functions}. 
\end{Def}

It is easy to see that test functions are closed under addition and scalar multiplication. The smoothness and compactness requirements are essential in the construction of distributions. Note that the smoothness requirement rules out most compactly supported functions. For our purposes there is a particular class of test functions which are the most useful ones and are easy to construct.

\begin{Ex}
Let $c \in \R$ and $r \in \R^+$. The \emph{bump function} $\varphi_r^c$ is the following:
\[\varphi_r^c = 
\begin{cases}
    e^{-\frac{1}{1 - (\frac{x-c}{r})^2}} & $if $ c - r < x < c + r\\
    0 & $otherwise$
\end{cases}
\]
\end{Ex}
Intuitively speaking, $\varphi_r^c$ is an infinitely differentiable bump of radius $r$ and centered at point $c$. $\varphi_1^0$ is pictured in figure \ref{fig:bump}. Additionally, because bump functions are closed under multiplication, we can easily extend this definition to the multivariate case:

\begin{Ex}
Let $c = (c_1, c_2, \dots, c_n) \in \R^n$ and $r \in \R^+$. The \emph{multivariate bump function} $\varphi_r^c:\R^n \to \R$ is the following:
\[
\varphi_r^c(x_1, x_2, \dots, x_n) = \varphi_r^{c_1}(x_1)\varphi_r^{c_2}(x_2)\dots\varphi_r^{c_n}(x_n)
\]
\end{Ex}

The set of test functions when equipped with function addition and multiplication by a scalar forms a vector space. More specifically, this means that we can normalize the volume of the image of the test function regardless of the radius or degree of differentiation. 


\begin{figure*}
\begin{multicols}{2}
    \includegraphics[width=\linewidth]{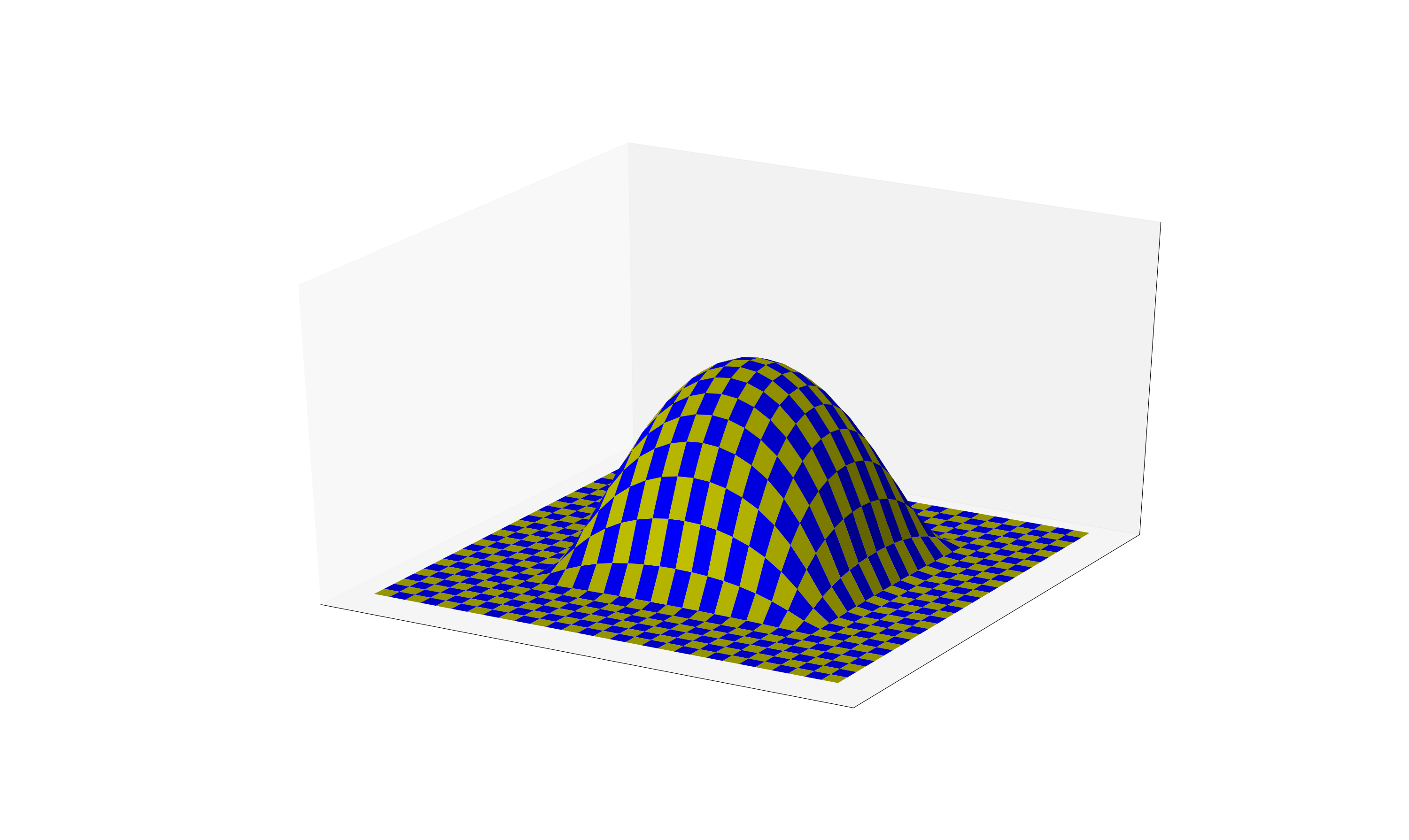}\par 
    \caption{Bump function centered at (0,0)}
    \label{fig:bump}
    \includegraphics[width=\linewidth]{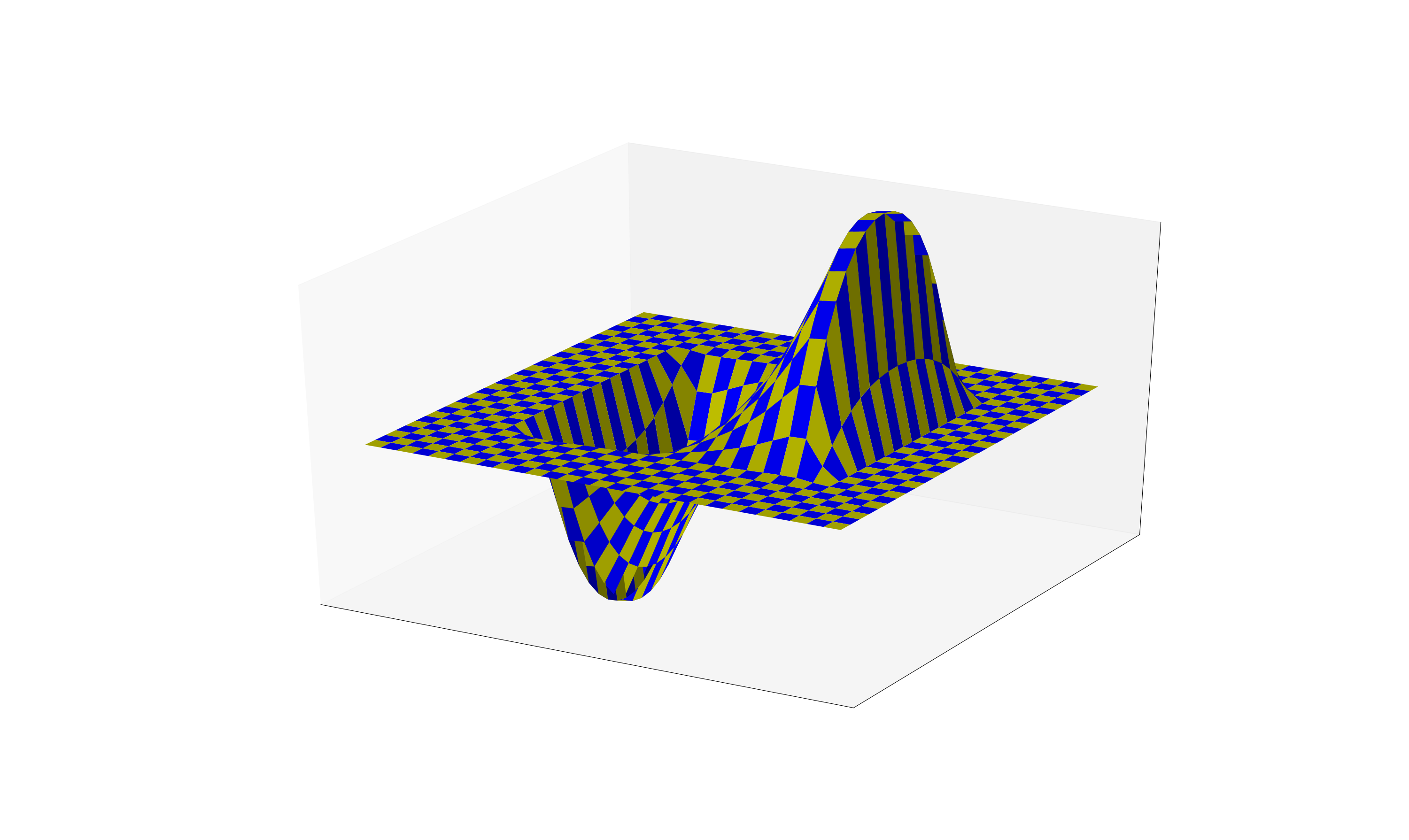}\par 
    \caption{First derivative of a bump function}
    \label{fig:deriv-bump}
    \end{multicols}
\end{figure*}

With these definitions in mind we define distributions as the set of continuous linear functionals $\dist U = \mathcal D(U) \multimap \R$. In the literature it is common to use the letters $T$ and $S$ to represent distributions, and  $\langle T, \varphi \rangle$ to denote the application of a distribution $T$ to a test function $\varphi$. It is important to note that the test functions having compact support is fundamental when defining derivatives of distributions, as illustrated by Theorem~\ref{th:loc}. The set of distributions over an open set forms a vector space, where addition is defined as $\langle T_1 + T_2, \varphi \rangle = \langle T_1, \varphi \rangle + \langle T_2, \varphi \rangle$ and scalar multiplication is defined as $\langle \alpha T, \varphi \rangle = \alpha \langle T, \varphi\rangle$. This construction will be used in Section~\ref{sec:syntax} to encode if-statements.

There are two particular families of distributions that should be mentioned, as they are the most important ones for our semantics. 

\begin{Ex}
  Let $u \in U$, we define the \emph{Dirac delta distribution} $\delta_u : \dist U$ as $\lrang {\delta_u, \varphi} = \varphi(u)$
\end{Ex}

This next class of distributions is why we call them ``generalized functions'', as every sufficiently nice function into $\R$ can be lifted to a distribution:

\begin{Th}[\cite{golse2010}]
  \label{th:lift}
  Let $f : U \to \R$ be a locally integrable function, then $\lrang {T_f, \varphi} = \int_U f(x)\varphi(x) d x$ is a distribution. Note that this definition uses the Lebesgue integral.
\end{Th}

\begin{Rem}
  Something peculiar about distributions is that we lose the ability to compute their values at a point, as they must be fed a continuation to output a real number. However, this is not too big a problem as we can use \emph{bump functions} (c.f. Figure~\ref{fig:bump} for an example) centered around a point with support contained in a sphere of radius $\varepsilon$ which approximates to an arbitrary precision the value of a function at a point.
\end{Rem}

\begin{Th}
\label{th:loc}
  Let $f : U \to \R$ be a locally integrable function, $u \in U$ a point where $f$ is continuous and $\psi_\varepsilon$ a family of positive test functions with volume $1$ such that its support is included in the $u$-centered $\varepsilon$-radius sphere, then $\lim_{\varepsilon \to 0} \lrang {T_f, \psi_\varepsilon} = f(u)$. 
\end{Th}

\subsubsection{Differentiation}
\label{sec:diff}

For every distribution $T \in \dist U$, where $U \subseteq \R^n$ is an open set we can define its partial derivative as the distribution $\frac{\partial T}{\partial x_i}$ such that $\lrang {\frac{\partial T}{\partial x_i}, \varphi} = - \lrang{T, \frac{\partial \varphi}{\partial x_i}}$. As such, distributional derivatives are heavily reliant on the derivatives of bump functions, the first of which is pictured in $\ref{fig:deriv-bump}$. 

We can show that this definition of differentiation extends usual differentiation in the following sense:

\begin{Th}
\label{th:smoothdist}
  For every differentiable function $f$, $\frac{\partial T_f}{\partial x_i} = T_{\frac{\partial f}{\partial x_i}}$.
\end{Th}

\begin{proof}
  The proof follows by using integration by parts and the fact that the test functions have compact support.
\end{proof}

This definition allows us to understand how distribution theory can be used to interpret the derivative of if-statements. Consider the following example:

\begin{Ex}
  Let $H(x) =  \mathsf{if }\, x < 0\, \mathsf{ then \,  } 0 \, \mathsf{ else }\, 1$ be the \emph{Heaviside function}. One can easily show using the definition above that $\frac{\mathrm d T_H}{\mathrm d x} = \delta_0$.
\end{Ex}

\begin{Ex}
  The function 
  $$f(x) = \mathsf{if }\, x < 0 \, \mathsf{ then  } \, 0 \, \mathsf{ else }\, x$$
  is the ReLU function frequently used as the activation function in neural networks. This function is differentiable almost everywhere except at $x = 0$. However, we can show that its distributional derivative is the Heaviside function defined above. Additionally, because bump functions are symmetric over their centered point, it is easy to see that the distributional derivative centered at $0$ of the ReLU function is $\frac{1}{2}$ regardless of the radius of the region tested. 
\end{Ex}

It is also important to note that, when dealing with nice piecewise continuous functions, the distributional derivative is somewhat "stable" with regards to the non-distributional derivative. For example, in piecewise continuous functions $f: \R \to \R$, the symbolic distributional derivative of the distribution $f$ is nearly identical to the symbolic piecewise derivative of the function $f$ save for Dirac delta distributions being injected at the points of discontinuity scaled by the magnitude of the discontinuity.

\begin{Th}[\cite{golse2010}]
  Let $f : \R \to \R$ be a piecewise differentiable function with discontinuities:
  $$
  a_1 < a_2 < \cdots < a_n
  $$

  The distributional derivative of $f$ is given by

  $$
  f' = f'|_{\R \backslash\{ a_1, \cdots, a_n\} } + \sum_{k = 1}^n \left (\lim_{x \to a_k^+} f(a_k) - \lim_{x \to a_k^-} f(a_k) f(a_k) \right) \delta_{a_k},
  $$
\end{Th}

When dealing with functions with a higher-dimensional domain the situation is a bit subtler but, under certain conditions, it can still be done analytically. 

\section{Typing rules}

\begin{figure*}
  \centering

\begin{tabular}{c}
  \begin{mathpar}

    \inferrule[Variable]{(x, \tau) \in \Gamma}{\Gamma \vdash x : \tau}
    
    \inferrule[Unpair]{\Gamma \vdash t_1 : \tau_1 \times \tau_2 \\ \Gamma, x:\tau_1, y:\tau_2 \vdash t_2 : \tau_3}{\Gamma \vdash \lettext (x, y) = t_1 \text{ in } t_2 : \tau_3}

\and

  \inferrule[Pair]{\Gamma \vdash t_1 : \tau_1 \\ \Gamma \vdash t_2 : \tau_2}{\Gamma \vdash (t_1, t_2) : \tau_1 \times \tau_2}

  \and



\inferrule[$\lambda$-Abstraction]{\Gamma, x:\tau_1 \vdash t:\tau_2}{\Gamma \vdash \lambda x:\tau_1.t : \tau_1 \xrightarrow{} \tau_2}

\and

\inferrule[Application]{\Gamma \vdash t_1 : \tau_1 \xrightarrow{} \tau_2 \\ \Gamma \vdash t_2 : \tau_1}{\Gamma \vdash t_1\ t_2 : \tau_2}

\and

\inferrule[Lift]{\Gamma \vdash t : \R^n \xrightarrow{} \R}{\Gamma \vdash \lift (t) : \dist {\R^n}}

\and

\inferrule[Indicator Function]{\Gamma \vdash t_1 : \mathsf{Pred}(\R^n) \\ \Gamma \vdash t_2 : \R^n \to \R }{\Gamma \vdash \mathbb{1}_{t_1}(t_2) : \dist{\R^n}}

\and

\inferrule[Differentiation]{\Gamma \vdash t : \dist {\R^n} \\ i \in \{1, ..., n\}}{\Gamma \vdash \frac{\partial t}{\partial x_i}:\dist {\R^n}}

\and

\inferrule [Iteration] {\Gamma \vdash t_0 : \tau \\ \Gamma \vdash t: \tau \to \tau}{\Gamma \vdash \iter t_0\ t : \nat \to \tau}

\and

\inferrule[Distribution Addition]{\Gamma \vdash t_1 : \dist {\R^n} \\ \Gamma \vdash t_2 : \dist {\R^n}}{\Gamma \vdash t_1 +. t_2 : \dist {\R^n}}

\and

\inferrule[Scalar Multiplication]{\Gamma \vdash \alpha : \R \\ \Gamma \vdash t_2 : \dist {\R^n}}{\Gamma \vdash \alpha t_2 :\dist {\R^n}}

\and

\inferrule[Distribution Application]{\Gamma \vdash t_1 : \dist {\R^n} \\ \Gamma \vdash t_2 : \test{\R^n}}{\Gamma \vdash \langle t_1, t_2 \rangle :\R}

\and

\inferrule[Bump Function]{\Gamma \vdash c : \R^n \\ \Gamma \vdash r : \R^+ \\ n \in \nat}{\Gamma \vdash \varphi^n(c, r):\test{\R^n}}

\and

\inferrule[Dirac delta]{\Gamma \vdash t : \R^n }{\Gamma \vdash \delta_t: \dist {\R^n}}

\end{mathpar}
\end{tabular}

\caption{Typing rules}
\label{fig:typing}
\end{figure*}







        




        












  
\end{document}